\documentclass[10pt,twoside,journal]{IEEEtran}

\usepackage{graphicx}
\usepackage{amsmath}
\usepackage{amsfonts}
\usepackage{latexsym}
\usepackage{amssymb}
\usepackage{amsthm}
\usepackage{color}
\usepackage{subfloat}
\usepackage{subfigure}
\usepackage{url}
\usepackage{mathrsfs}   
\usepackage[justification=raggedright]{caption}
\usepackage{pifont}
\usepackage{cite}

\usepackage{array}
\newcommand{\PreserveBackslash}[1]{\let\temp=\\#1\let\\=\temp}
\newcolumntype{C}[1]{>{\PreserveBackslash\centering}p{#1}}
\newcolumntype{R}[1]{>{\PreserveBackslash\raggedleft}p{#1}}
\newcolumntype{L}[1]{>{\PreserveBackslash\raggedright}p{#1}}
\usepackage{tabularx}

\newtheorem{theorem}{\bf Theorem}
\newtheorem{lemma}[theorem]{\bf Lemma}

\newtheorem{definition}{\bf Definition}
\newtheorem{problem}{\bf Problem}
\newtheorem{observation}{\bf Observation}

\begin{document}

\title{REAP: An Efficient Incentive Mechanism for Reconciling Aggregation Accuracy and Individual Privacy in Crowdsensing}

\author{Zhikun Zhang,~\IEEEmembership{Student member,~IEEE,}
        Shibo He,~\IEEEmembership{Member,~IEEE,} \\
        Jiming Chen,~\IEEEmembership{Senior Member,~IEEE,}
        and~Junshan~Zhang,~\IEEEmembership{Fellow,~IEEE,}
\IEEEcompsocitemizethanks{\IEEEcompsocthanksitem Z. Zhang, S. He and J. Chen (Corresponding author) are with State Key Laboratory of Industrial Control Technology, Zhejiang University, and Cyber Innovation Joint Research Center, Hangzhou, China. \protect
E-mail: zhangzhk@zju.edu.cn, s18he@iipc.zju.edu.cn, cjm@zju.edu.cn}
\IEEEcompsocitemizethanks{\IEEEcompsocthanksitem Junshan. Zhang is School of Electrical, Computer and Energy Engineering, Arizona State University, USA. \protect
E-mail: junshan.zhang@asu.edu}
}


\IEEEtitleabstractindextext{
\begin{abstract}
Incentive mechanism plays a critical role in privacy-aware crowdsensing. Most previous studies on co-design of incentive mechanism and privacy preservation assume a trustworthy fusion center (FC). Very recent work has taken steps to relax the assumption on trustworthy FC and allows participatory users (PUs) to add well calibrated noise to their raw sensing data before reporting them, whereas the focus is on the equilibrium behavior of data subjects with binary data. Making a paradigm shift, this paper aim to quantify the privacy compensation for continuous data sensing while allowing FC to directly control PUs. There are two conflicting objectives in such scenario: FC desires better quality data in order to achieve higher aggregation accuracy whereas PUs prefer adding larger noise for higher privacy-preserving levels (PPLs). To achieve a good balance therein, we design an efficient incentive mechanism to REconcile FC's Aggregation accuracy and individual PU's data Privacy (REAP). Specifically, we adopt the celebrated notion of differential privacy to measure PUs' PPLs and quantify their impacts on FC's aggregation accuracy. Then, appealing to Contract Theory, we design an incentive mechanism to maximize FC's aggregation accuracy under a given budget. The proposed incentive mechanism offers different contracts to PUs with different privacy preferences, by which FC can directly control PUs. It can further overcome the \emph{information asymmetry}, i.e., the FC typically does not know each PU's precise privacy preference. We derive closed-form solutions for the optimal contracts in both \emph{complete information} and \emph{incomplete information} scenarios. Further, the results are generalized to the continuous case where PUs' privacy preferences take values in a continuous domain. Extensive simulations are provided to validate the feasibility and advantages of our proposed incentive mechanism.
\end{abstract}


}

\maketitle

\IEEEdisplaynontitleabstractindextext

\IEEEpeerreviewmaketitle

\maketitle \thispagestyle{empty}


\section{Introduction}\label{intr}
\IEEEPARstart{T}{he} recent proliferation of portable mobile devices (e.g., smartphone, smartwatch, tablet computer, etc.), integrated with a set of sensors (e.g., GPS, camera, accelerometer, etc.), has spurred much interest in mobile crowdsensing \cite{duan2016distributed,he2016near}. Due to its advantage in reducing the deployment cost in large-scale sensing applications, crowdsensing has been applied to a large variety of areas such as smart transportation, environmental monitoring, health-care, etc \cite{ganti2010greengps,gao2014jigsaw,cheng2014aircloud,hu2015smartroad}.

Typically, sensing data collected from participatory users (PUs) will be aggregated by the fusion center (FC) for data analytics. To identify public health condition, for example, FC can collect the daily exercise data from PUs and carry out data aggregation such as average and histogram. Clearly, contributing sensing data to FC is costly for PUs, since resources such as energy and bandwidth will be consumed and data privacy may be sacrificed. Therefore, they would be reluctant to participate in crowdsensing without a proper incentive mechanism that compensates their cost. Most previous studies focused on resources consumption for data sensing and reporting in incentive mechanism design \cite{luo2014profit,zhang2015incentivize,zhang2015truthful}. Only quite a few consider PUs' privacy losses\cite{jin2016inception,zhang2016privacypreserving} and common assumption made by these works is that FC is trustworthy such that privacy merely breaches when FC releases the aggregation results to the public.

In reality, the trustworthy FC assumption may not hold, e.g., when FC is compromised by malicious attackers, or the communication channels between PUs and FC are eavesdropped. Very recent work \cite{wang2016strategic} take the first attempt to remove the trustworthy authority assumption and study how to trade private data in a game-theoretic model. In \cite{wang2016strategic}, PUs can fully control their privacy by adding well calibrated noise to the raw data before reporting them. However, the private data is assumed to be binary, which is not often times applicable to real-world system. Further, the focus of \cite{wang2016strategic} is on examining the equilibrium behavior of data subjects such that data collector have no direct control of them. Different from \cite{wang2016strategic}, this paper aim to quantify the privacy compensation for continuous data sensing while allowing FC to directly control PUs.

One challenge in doing this is to reconcile the following conflict: PUs prefer adding larger noise for higher privacy preserving levels (PPLs) whereas FC desires better quality data for higher aggregation accuracy. Another challenge is to overcome the \emph{information asymmetry} problem between FC and PUs, since it is difficult (perhaps impossible) to know PUs' privacy preferences. Further, privacy preferences of PUs are typically heterogenous, e.g., women have higher privacy preferences about their age than men, and patients are more concerned about their location privacy, which incur diverse privacy losses for different PUs under the same PPL. An efficient incentive mechanism needs to differentiate the diverse privacy losses of PUs and provide appropriate rewards that capture their contribution to FC without knowing individual PU's precise privacy preference.

To tackle these challenges, we propose REAP\footnote{The name REAP comes from \underline{RE}conciling \underline{A}ggregation accuracy and individual \underline{P}rivacy.}, an efficient incentive mechanism based on Contract Theory. By Contract Theory, FC can add some kind of enforcement to incentivize PUs by signing specific contracts with them, so that FC has direct control over PUs. Different contracts should be designed for different types of PUs, each of which specifies one type of PPL and the corresponding payment that a PU will receive if he/she can sacrifice the given PPL. A key concern here is to design a proper menu of contracts satisfying incentive compatibility such that all PUs can maximize their utilities only when they truthfully reveal their privacy preferences.

Specifically, we adopt differential privacy to quantify individual privacy and $(\alpha,\delta)$-accuracy to measure FC's aggregation accuracy. Then, the quantitative relationship between individual PU's PPL and FC's aggregation accuracy is derived. In light that the contribution of each PU to the aggregation accuracy can be quantified, we design a menu of optimal contracts that maximize FC's aggregation accuracy under a given budget. We first consider the \emph{complete information} scenario as a benchmark, where FC knows the precise type of each PU. This benchmark serves as the best aggregation accuracy that FC can achieve. We further consider the optimal contract design in \emph{incomplete information} scenario where FC only knows the probability distribution of PUs' types. Closed-form solutions for both scenarios are derived. Further, we generalize our results to the continuous case where PUs' privacy preferences can take value in a continuous domain. In such a case, the optimization problem turns out to be a functional extreme value problem that can be solved by an optimal control based approach.

The contributions of this paper are there folds:
\begin{enumerate}
  \item We propose REAP, a Contract Theory based incentive mechanism, to compensate PUs' data privacy losses and hence resolve the information asymmetry issues between PUs and FC.
  \item We adopt proper measures to quantify both individual PUs' PPLs and FC's aggregation accuracy, by which the quantitative relationship between individual privacy and aggregation accuracy is derived.
  \item Closed-form solutions are derived for both complete information and incomplete information scenarios. We also generalize our results to the case of continuous privacy preferences.
\end{enumerate}

The rest of this paper is organized as follows. The related work is discussed in Section \ref{section:relatedwork}. Section \ref{section:systemmodel} presents an overview to the crowdsensing system, and quantify PUs' PPLs as well as their impacts on FC's aggregation accuracy. In Section \ref{section:contractformulation}, we leverage Contract Theory to address the information asymmetry problem and generalize our results to the continuous case in Section \ref{section:continuum}. Simulation results are illustrated in Section \ref{section:simulation} to validate our theoretical results. Section \ref{section:conclusion} concludes this paper.


\section{Related Work}\label{section:relatedwork}

%

Recently, various incentive mechanisms have been proposed to incentivize users's participation in moblie crowdsensing systems. Most of these mechanisms are based on either auction \cite{yang2015incentive,jin2016inception,zhang2016privacypreserving,jin2015quality,zhang2014free,koutsopoulos2013optimal} or other game-theoretic models \cite{duan2012incentive,luo2015crowdsourcing,peng2015pay,cheung2015distributed,xie2014incentive}, which aim to achieve different objectives. Specifically, in \cite{cheung2015distributed,jin2015quality}, the authors aim to maximize the social welfare. The objective of \cite{duan2012incentive,luo2015crowdsourcing} is to maximize the profit of the platform, and \cite{koutsopoulos2013optimal,xie2014incentive} design mechanisms to minimize FC's payment. The basic requirement of these mechanisms is to guarantee that all users' cost is compensated, at least in the expectation sense. Most previous studies only compensate users' resource consumption for sensing and reporting data, their privacy loss is not remunerated explicitly.

Interestingly, Ghosh et al. took the first step to view privacy as a good and aim to compensate users' privacy loss in their seminal work \cite{ghosh2011selling} in data mining field. In \cite{ghosh2011selling}, data owners bid their privacy loss based on their privacy preference, and the system chooses a set of users and the corresponding PPLs to achieve the best statistic accuracy under a given budget. Based on this work, a few improved mechanisms \cite{fleischer2012approximately,ligett2012take,nissim2014redrawing} have been proposed, especially consider the correlation between privacy preference and private data. Most of these mechanisms require a trustworthy authority, which is not available in most cases. Recently, Wang et. al. \cite{wang2016strategic} removed the trustworthy authority assumption in data mining field and proposed a game-theoretic approach to compensate users' privacy loss. However, the private data considered in \cite{wang2016strategic} is always binary bit, which is not widely applicable in mobile crowdsensing systems. Further, \cite{wang2016strategic} do not consider the information asymmetry problem between FC and PUs. Thus motivated, in this paper, we consider a more realistic crowdsensing scenario where the FC is untrusted and allow PUs to take full control of their private data, which take continuous value. Moreover, a novel incentive mechanism based on Contract Theory is proposed to handle the information asymmetry problem.

Another line of related work is privacy-preserving mechanism design in mobile crowdsensing systems. These works do not take users' data privacy into consideration. Instead, they consider the privacy issue of the mechanism itself. For example, \cite{li2014providing,li2013providing} aimed to preserve users' anonymity within the incentive mechanism, and \cite{jin2016enabling} aimed to preserve users' bid privacy.


\section{System Model}\label{section:systemmodel}
In this section, we first present the system overview. Then, we quantify PUs' PPLs and their impacts on FC's aggregation accuracy.

\subsection{System Overview}
The mobile crowdsensing system considered in this paper consists of an untrusted FC, a task agent and a set $\mathcal{U}=\{u_1,u_2,\cdots,u_n\}$ of PUs as shown in Fig. \ref{fig:systemmodel}. Different from most of the previous works on privacy-preserving data aggregation in crowdsensing, we remove the trustworthy FC assumption, since FC may be compromised by malicious attackers, or the communication channels between PUs and FC maybe eavesdropped.

The FC aims to collect a set of sensing data from $n$ PUs, denoted as $\mathcal{D}=\{d_1,d_2,\cdots,d_n\}$,  where $d_i \in \mathbb{R}$ is a real number. Then it carries out some aggregation operations, such as average, max/min, histogram, etc, to abstract some valuable patterns. For easy exposition, we will investigate the average aggregation\footnote{We leave the discussion of other kinds of data aggregations in future work}, i.e., $s=\frac{1}{n}\sum_{i=1}^nd_i$, which constitute a large portion of currently deployed crowdsensing system. For example, some map application such as Baidu map collect GPS data (e.g., location and speed) from mobile vehicles and conduct average aggregation to monitor the real-time traffic condition. In the healthcare application, FC intends to collect PUs' daily exercise data and conduct average aggregation to monitor public health condition.

Clearly, the sensing data may contain sensitive information about PUs. Abuse of these sensitive information may breach PUs' privacy. Considering the healthcare application, the exercise data allow adversaries to infer individual PU's health condition or living habit. Therefore, PUs may not be willing to contribute their raw sensing data due to the privacy concern. To dispel PUs' worry about privacy, we propose to allow for PUs to add well-calibrated noise $\eta_i$ to their raw sensing data $d_i$ before reporting them to the FC, and their PPLs can be strictly quantified by differential privacy as depicted in Section \ref{subsection:dp}. 
\begin{figure}[!ht]
\begin{center}
\includegraphics[width=0.35\textwidth]{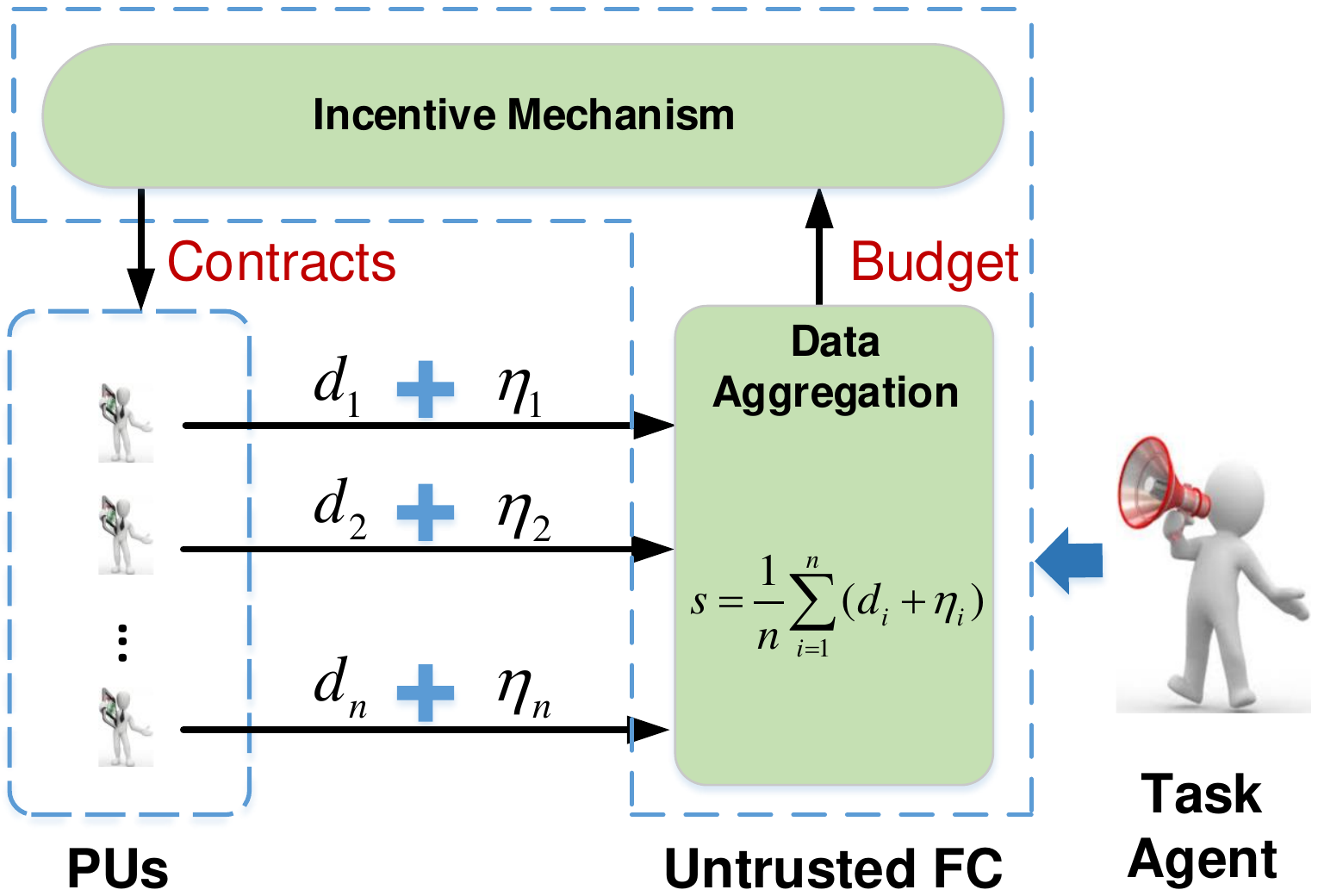}
\end{center}
\caption{Framework of REAP.} \label{fig:systemmodel}
\end{figure}

However, there are two conflicting objectives in this setting: FC desires better quality data in order to achieve higher aggregation accuracy whereas PUs prefer adding larger noise for higher PPLs (these conflicts will further be quantified in Section \ref{subsection:accuracy}). In this paper, we aim to design an efficient mechanism to reconcile these conflicts. The framework of the proposed crowdsensing system is shown in Fig. \ref{fig:systemmodel} and the workflow is as follows:
\begin{itemize}
  \item Firstly, the task agent announces a sensing task to the FC.
  \item \textbf{Incentive Mechanism.} Then, FC designs a menu of contract items (each specifies a privacy-payment pair) that maximize the aggregation accuracy under given budget, and broadcast them to all PUs. PUs can choose to sign any one of the contract that maximize their own utilities. Once the contract is signed, PUs must report a privacy-preserving version of their sensing data with the PPLs specified in the contracts. In return, they will receive the corresponding payments.
  \item \textbf{Data Aggregation.} Next, after receiving the privacy-preserving sensing data from PUs, FC conduct average aggregation on these data.
  \item Finally, FC return the aggregated data to the task agent.
\end{itemize}




\subsection{Differentially Private Data Reporting} \label{subsection:dp}
In this subsection, we adopt the celebrated notion of differential privacy \cite{dwork2008differential} to quantify individual PU's PPL and privacy loss, and then define PUs' utility function. 

Informally, differential privacy guarantees that, after receiving the observation, the attackers cannot distinguish the neighboring input with high confidence. Here, neighboring relationship is an important concept in differential privacy. In this paper, we adopt the neighboring relationship for continuous value as follows:
\begin{definition}[$\gamma_i$-adjacency]\label{definition:adjacency}
Two continuous data $d_i$ and $d_i'$ are $\gamma_i$-adjacency, if $|d_i-d_i'|\leq \gamma_i$, where $\gamma_i$ is the range of PU $i$'s sensing data $d_i$.
\end{definition}

Then, we can give the formal definition of differential privacy.
\begin{definition}[$\epsilon_i$-differential privacy \cite{le2014differentially}]\label{definition:dp}
A random algorithm $\{\mathcal{A}:R \rightarrow R|\mathcal{A}(d_i)=d_i+\eta_i\}$ achieves $\epsilon_i$-differential privacy, if for all pairs of $\gamma_i$-adjacency data $d_i$ and $d_i'$, and observation $d^{obs}$,
\begin{align}\label{formula:dp}
Pr[\mathcal{A}(d_i)=d^{obs}]\leq e^{\epsilon_i}Pr[\mathcal{A}(d_i')=d^{obs}].
\end{align}
\end{definition}

Intuitively, PU $i$'s accurate sensing data can be either $d_i$ or $d_i'$ from an attacker's view. After adding noise $\eta_i$, both $d_i$ and $d_i'$ can result in $d^{obs}$ with certain probability. Thus, an attacker cannot distinguish PU $i$'s accurate sensing data with high confidence when he observe $d^{obs}$. Clearly, smaller $\epsilon_i$ means higher PPL, since it is harder to distinguish $d_i$ and $d_i'$ when observing $d^{obs}$.

The Laplacian mechanism \cite{dwork2006calibrating} is the first and probably most widely used mechanism for achieving differential privacy, it satisfies $\epsilon_i$-differential privacy by calibrating the Laplacian noise parameter based on the following lemma:
\begin{lemma}\label{laplacian}
If the Laplacian mechanism is used, i.e., $\eta_i\sim Lap(0,b_i)$, we can achieve $\epsilon_i$-differential privacy if $b_i=\frac{\gamma_i}{\epsilon_i}$.
\end{lemma} 


By differential privacy, we can also define PUs' privacy loss. According to the utility theoretic characterization of differential privacy \cite{pai2013privacy}, the relationship between the expected utilities with two adjacent data can be characterized by $e^{\epsilon_i}$ based on (\ref{formula:dp}). Following \cite{ghosh2011selling}, the privacy loss can be modeled as the difference between the utility with true data and the utility with perturbed data, which is a linear function of $\epsilon_i$ when it is small. Since $e^{\epsilon_i} \approx 1 + \epsilon_i$ for small value of $\epsilon_i$. Then, we can define PUs' utility in Definition \ref{definition:pu_utility}.


\begin{definition}[PUs' utility] \label{definition:pu_utility}
Any $PU_i$'s utility is defined as
\begin{align}\label{formula:pu_utility}
u_i=p_i - \theta_i\epsilon_i, 
\end{align} 
\end{definition}

where $p_i$ is PU $i$'s reward when he/she contribute sensing data to FC. $\theta_i$ is the privacy preference of PU $i$ which indicate how much PUs care about their privacy. Clearly, different PUs may have different privacy preferences \cite{xu2015privacy}, for instance, patients in hospital have higher privacy preference to their location than others. Naturally, individual PU's privacy preference is private information and unknown to FC, or in other words, there exists \emph{information asymmetry} between FC and PUs.

Notice that we only consider the cost incurred by PUs' privacy loss in order to ease the presentation in this paper, meanwhile the result in this paper can be extended to incorporate the sensing cost. For instance, similar to \cite{zhang2016privacypreserving}, setting PU $i$'s sensing cost to $s_i$, we can rewrite PU $i$'s utility as $u_i = p_i - s_i - c_i$ and define $p_i' = p_i - s_i$ to incorporate the sensing cost in the payment.

\subsection{Privacy versus Accuracy} \label{subsection:accuracy}
In this subsection, we illustrate the conflicts between FC's aggregation accuracy PUs' PPLs by deriving their quantitative relationship. 

To quantify the aggregation accuracy of the privacy preserving sensing data, we adopt the following accuracy definition.
\begin{definition}[$(\alpha,\delta)$-accuracy]
The aggregation $\hat{s}$ of privacy-preserving sensing data achieves $(\alpha,\delta)$-accuracy if
\begin{align*}
Pr[|\hat{s}-s|\geq\alpha] \leq 1-\delta,
\end{align*}
\end{definition}
where $s$ is the aggregation result of accurate sensing data.

Intuitively, this definition indicates that the aggregation error is larger than $\alpha$, with probability at most $1-\delta$. From estimation's perspective, $\alpha$ stands for confidence interval and $\delta$ stands for confidence level. Clearly, for a given confidence level, a smaller confidence interval means better aggregation accuracy. Thus, we can leverage the confidence interval $\alpha$ under a certain confidence level to measure the \emph{aggregation accuracy}, where a smaller $\alpha$ means better aggregation accuracy.

Then, we derive the quantitative relationship between individual PU's privacy and FC's aggregation accuracy as the following lemma:
\begin{lemma}\label{lemma:relationship}
For a given confidence level $\delta\leq1$, the aggregation accuracy $\alpha$ of the privacy-preserving sensing data can be found as
\begin{align}\label{formula:relationship}
\alpha = \frac{\sqrt{2}\gamma}{n\sqrt{1-\delta}}\sqrt{\sum_{i=1}^n\frac{1}{\epsilon_i^2}}.
\end{align}
\end{lemma}
where $\epsilon_i$ is PU $i$'s PPL, $n$ is the number of PUs, and $\gamma$ is the range of PUs' sensing data\footnote{Notice that the range of the sensing data should be the same for all PUs in a specific crowdsensing application, for example, the heart rate of a normal adult is always in the range $60\sim100$ bpm. Thus, all PUs' $\gamma_i$ should take the same value, i.e., $\gamma_i=\gamma,\forall\gamma_i$.}. The proof can be found in Appendix \ref{appendix:1}.

Recall that a smaller $\epsilon_i$ and $\alpha_i$ means higher PPL and aggregation accuracy, by examining Formula (\ref{formula:relationship}), we can see that the FC and PUs have conflicting objectives. The FC wants PUs to adopt lower PPLs, which increases FC's aggregation accuracy. PUs want to adopt higher PPLs to better preserve their privacy, which decrease FC's aggregation accuracy. In the next section, we resolve this conflict through Contract Theory. 
\section{Incentive mechanism design: A Contract Theoretic Approach} \label{section:contractformulation}
So far, we have quantified the conflicts between PUs' privacy and FC's aggregation accuracy. In this section, we introduce the contract mechanism to resolve the conflicting objectives between PUs and FC.

\subsection{Contract Formulation}

Contract theory generally studies how economic decision-makers construct contractual arrangement in the presence of \emph{information asymmetry}, i.e., FC typically does not know each PUs' privacy preference $\theta_i$, and aim to design a menu of contracts to incentivize PUs to participate in crowdsensing to maximize the aggregation accuracy. To facilitate later discussion, we classify PUs into different types based on their privacy preferences, i.e., the privacy preference of type-i PUs is $\theta_i$.

In this section, we consider the case where PUs have finite types of privacy preference, say $k$ types $\Theta=\{\theta_i, \theta_2, \cdots\, \theta_k\}$, and provide some insight to the contract design. We leave the discussion of the case where $\theta$ takes continuous value in the next section. To facilitate the analysis, we sort PUs' types in ascending order, i.e., $\theta_1\leq\theta_2\leq\cdots\leq\theta_k$, i.e., a higher type of PU has a higher privacy preference. Using Contract theory, FC designs a contract that specifies the relationship between a PU's PPL $\epsilon_i$ and the corresponding payment $p_i$ that a PU will receive if he/she can sacrifice the given PPL. Specifically, a contract is a set $\mathcal{C} = \{(\epsilon_1, p_1), \cdots, (\epsilon_k, p_k)\}$ of privacy-payment pairs called contract items. Each PU choose to sign a contract item $(\epsilon_i, p_i)$ and report $\epsilon_i$-differentially private sensing data for the payment $p_i$. Once the contract is signed, a PU must report a privacy-preserving version of sensing data and FC must reward him according to the item.

Each type of PUs choose the contract item that maximizes their utilities in (\ref{formula:pu_utility}). FC aims to optimize the contract and maximize the aggregation accuracy, i.e., minimize $\alpha$ in (\ref{formula:relationship}). Since $\frac{\sqrt{2}\gamma}{n\sqrt{1-\delta}}$ is a positive constant, minimizing $\alpha = \frac{\sqrt{2}\gamma}{n\sqrt{1-\delta}}\sqrt{\sum_{i=1}^n\frac{1}{\epsilon_i^2}}$ is equivalent to minimize $\alpha = \sum_{i=1}^n\frac{1}{\epsilon_i^2}$.

In the following subsection, we will consider the optimal contract design under two information scenarios.
\begin{itemize}
  \item \textbf{Complete information:} The complete information scenario is served as a benchmark, where FC knows each PU's precise type, and can offer a specific contract to each PU directly. Clearly, FC can achieve the best aggregation accuracy in this scenario, which serves as the upper bound of FC's achievable aggregation accuracy in any information scenario.
  \item \textbf{Incomplete information:} In the incomplete information scenario, the FC do not know each PU's precise type, but know the distribution of each type, e.g., type-$i$ has $\lambda_i$ PUs. In this scenario, FC should decide and broadcast a menu of optimal contracts to all PUs, and each PU can choose the contract that maximize his/her utility.
\end{itemize}

\subsection{Optimal Contract Design under Complete Information}
In the complete information scenario, FC knows each PU's precise type. We will leverage the optimal aggregation accuracy achieved in this case as a benchmark to evaluate the performance of the proposed contract under incomplete information scenario. As FC knows each PU's type, it can offer a specific contract to each PU directly. In this scenario, FC only need to guarantee that each PU's utility is nonnegative so that they are willing to contribute their sensing data. In Contract Theory, we call this individual rationality constraint.
\begin{definition}[Individual Rationality] \label{definition:ir}
A menu of \\contracts satisfy Individual Rationality (IR) constraint if they provide nonnegative utility to all PUs, i.e., 
\begin{align}
p_i-\theta_i\epsilon_i\geq0,\forall i.
\end{align}
\end{definition}

Thus, we can design the optimal contract under complete information by solving the following optimization problem:
\begin{problem}\label{problem:systemgoalcom}
\begin{align}
&\min\ \sum_{i=1}^k\frac{\lambda_i}{\epsilon_i^2}, \nonumber \\
&s.t.\ \sum_{i=1}^k\lambda_ip_i\leq B, \label{formula:budget}\\
&\ \ \ \ p_i-\theta_i\epsilon_i\geq0, \ \ \forall i. \label{formula:IR}
\end{align}
\end{problem}
where $B$ is the total budget that FC possesses.

Then, we provide the solution to this optimization problem.
\begin{lemma}\label{lemma:equality}
The inequality in (\ref{formula:budget})(\ref{formula:IR}) can take the equal sign simultaneously, i.e., $\sum_{i=1}^k\lambda_ip_i=B$ and $p_i-\theta_i\epsilon_i=0$.
\end{lemma}

It is easy to show that both (\ref{formula:budget}) and (\ref{formula:IR}) can take the equal sign by contradiction. Given $p_i$, if there exists an optimal contract that satisfies $p_i-\theta_i\epsilon_i>0$, then we can always find a larger $\epsilon_i$ to achieve better aggregation accuracy until the equality satisfies. Similarly, If there exists an optimal contract that satisfies $\sum_{i=1}^k\lambda_ip_i<B$, we can always find a larger $p_i$, which means larger $\epsilon_i$, to achieve better aggregation accuracy until the equality satisfies, which lead to the correctness of this lemma.

Lemma \ref{lemma:equality} shows that both IR constraints and budget constraint are tight at the optimal solution to Problem (\ref{problem:systemgoalcom}), which indicate that the FC can provide a zero utility to each type-i PU with $p^* = \theta_i \epsilon^*$ and spend all the feasible budget. Therefore, Problem \ref{problem:systemgoalcom} can be reduced to the following problem:
\begin{problem}\label{problem:systemgoalcom1}
\begin{align}
&\min\ \sum_{i=1}^k\frac{\lambda_i}{\epsilon_i^2}, \nonumber \\
&s.t.\ \sum_{i=1}^k\lambda_ip_i=B, \label{formula:budget1}\\
&\ \ \ \ p_i-\theta_i\epsilon_i=0, \ \ \forall i. \label{formula:IR1}
\end{align}
\end{problem}

By solving Problem \ref{problem:systemgoalcom1}, we have the following theorem.

\begin{theorem}\label{theorem:complete}
In the complete information scenario, the optimal contract $\{\epsilon_i^*,p_i^*\}$ is given by
\begin{align}
\epsilon_i^*&=\frac{B}{\sum_{j=1}^k\lambda_j\theta_j^{\frac{2}{3}}}\theta_i^{-\frac{1}{3}}, \label{formula:optcontractcomepsilon} \\
p_i^*&=\frac{B}{\sum_{j=1}^k\lambda_j\theta_j^{\frac{2}{3}}}\theta_i^{\frac{2}{3}}. \label{formula:optcontractcomp}
\end{align}
\end{theorem}

The proof can be found in Appendix \ref{appendix:2}. By looking into the parameters in the optimal contract provided in Theorem \ref{theorem:complete}, we have the following observation.

\begin{observation}
Recall that a smaller $\epsilon_i$ means higher PPL, Theorem \ref{theorem:complete} shows that the PPL to a type-$i$ PU decreases in $B$, and increases in $\theta_i$, which conforms to our intuition. That is, more budget can incentivize PUs to choose lower PPLs to achieve higher aggregation accuracy, and FC tends to buy less privacy from PUs with higher privacy preference to reduce payment.
\end{observation}

\subsection{Optimal Contract Design under Incomplete Information}
In the incomplete information scenario, FC does not know each PU's precise type, while the distribution of PUs' types is assumed to be known, i.e., type-$i$ have $\lambda_i$ PUs. In practice, the distribution of PUs' types can be obtained through questionnaire survey or analysis of the historical behavior of PUs \cite{zhang2015contract,duan2014cooperative}. Clearly, FC should design an optimal contract for each type of PUs to achieve best accuracy, but due to the lack of knowledge about each PU's precise type, FC can only broadcast all contracts to all PUs. However, if choosing the contract designed for other types can bring them higher utilities, some selfish PUs may pretend to be other types. To encourage all PUs to truthfully reveal their types, the optimal contracts should guarantee that choosing the contract corresponding to their own type can always achieve the highest utilities. Formally, we define this requirement as incentive compatibility constraint.
\begin{definition}[Incentive Compatibility]
A menu of contracts satisfies Incentive Compatibility (IC) constraint if the contract designed for type-$i$ PUs brings them the highest utility, i.e., 
\begin{align}\label{formula:IC}
p_i-\theta_i\epsilon_i\geq p_j-\theta_i\epsilon_j, \ \ \forall j\neq i.
\end{align}
\end{definition}

Apart from the incentive compatibility constraint, the contract under incomplete information should also satisfy the individual rationality constraint in Definition \ref{definition:ir}. Thus, we can design the optimal contract under incomplete information by solving the following optimization problem:
\begin{problem}\label{problem:systemgoalincom}
\begin{align}
&\min\ \sum_{i=1}^k\frac{\lambda_i}{\epsilon_i^2}, \nonumber \\
&s.t.\ \sum_{i=1}^k\lambda_ip_i\leq B, \label{formula:budget2} \\
&\ \ \ \ p_i-\theta_i\epsilon_i\geq0, \ \ \forall i, \label{formula:IR2} \\
&\ \ \ \ p_i-\theta_i\epsilon_i\geq p_j-\theta_i\epsilon_j, \ \ \forall j\neq i.
\end{align}
\end{problem}

In Problem \ref{problem:systemgoalincom}, there are $k$ IR constraints and $k(k-1)$ IC constraints, which makes it difficult to solve the optimization problem. Next, we show that these constraints can be reduced to a set of fewer equivalent constraints by the following lemmas.
\begin{lemma}\label{lemma:reduceIR}
The $k$ IR constraints can be reduced to the following one constraint:
\begin{align}
p_k-\theta_k\epsilon_k=0. \label{formula:reducedIR}
\end{align}
\end{lemma}

\begin{proof}
Notice that we have sort PUs' type in ascending order, i.e., $\theta_1\leq\theta_2\leq\cdots\leq\theta_k$, and based on IC constraint, we have
\begin{align*}
p_i-\theta_i\epsilon_i\geq p_k-\theta_i\epsilon_k\geq p_k-\theta_k\epsilon_k,\forall i\neq k.
\end{align*}

Thus, if the IR constraint of type-$k$ satisfied, i.e., $p_k-\theta_k\epsilon_k\geq0$, it will satisfied for all other types automatically. Therefore, we can keep the last IR constraint and reduce the others. Moreover, if there exists an optimal contract that satisfies $p_k-\theta_k\epsilon_k>0$, we can always find a larger $\epsilon_k$ to achieve better aggregation accuracy until $p_k-\theta_k\epsilon_k=0$, which end the proof.
\end{proof}

Lemma \ref{lemma:reduceIR} shows that only the highest type of PUs receive a zero utility, and lower types of PUs receive positive utilities that are decreasing in their types. The reason is that FC does not know each PU's type, it needs to provide incentives in terms of positive utilities to PUs to attract them revealing their truthful types. This is called \emph{information loss} compared to complete information.

\begin{lemma}[Monotonic Property]\label{lemma:monotonic}
If $\theta_1\leq\theta_2\leq\cdots\leq\theta_k$, then $\epsilon_1\geq\epsilon_2\geq\cdots\geq\epsilon_k$ holds.
\end{lemma}

\begin{proof}
Based on the IC constraint, we have
\begin{align*}
p_i-\theta_i\epsilon_i\geq p_j-\theta_i\epsilon_j,\\
p_j-\theta_j\epsilon_j\geq p_i-\theta_j\epsilon_i.
\end{align*}

Adding these two inequalities result in
$\epsilon_i(\theta_j-\theta_i)\geq\epsilon_j(\theta_j-\theta_i)$. Thus, we have if $\theta_i\leq\theta_j$, then $\epsilon_i\geq\epsilon_j$ for all $i$ and $j$, which lead to the correctness of this lemma.
\end{proof}

Intuitively, Lemma \ref{lemma:monotonic} shows that a PU with higher type should be assigned lower PPL, since his unit cost is higher and the FC needs to compensate this PU more when the contribution to the aggregation accuracy are the same. Further, this Lemma can be leveraged to prove the correctness of Lemma \ref{lemma:reduceIC}.

\begin{lemma}\label{lemma:reduceIC}
The $k(k-1)$ IC constraints can be reduced to the following $k - 1$ constraints.
\begin{align}
p_i-\theta_i\epsilon_i=p_{i+1}-\theta_i\epsilon_{i+1},\forall i \leq k - 1. \label{formula:reducedIC}
\end{align}
\end{lemma}

The proof can be found in Appendix \ref{appendix:3}. Lemma \ref{lemma:reduceIC} ensures that if the contract item $(\epsilon_i, p_i)$ designed for type-$i$ PUs bring them the same utilities with the contract item $(\epsilon_{i+1}, p_{i+1})$ designed for type-$(i+1)$ PUs, all the IC constraints for type-$i$ PUs are satisfied, which means type-$i$ PUs will truthfully select the contract item designed for their corresponding type.

Based on Lemma \ref{lemma:reduceIR} and Lemma \ref{lemma:reduceIC}, we can reduce Problem \ref{problem:systemgoalincom} to the following problem:
\begin{problem}\label{problem:systemgoalincom1}
\begin{align}
&\min\ \sum_{i=1}^k\frac{\lambda_i}{\epsilon_i^2}, \nonumber \\
&s.t.\ \ \sum_{i=1}^k\lambda_ip_i=B, \label{formula:psysincom1_1} \\
& \ \ \ \ \ p_k-\theta_k\epsilon_k=0, \label{formula:psysincom1_2} \\
& \ \ \ \ \ p_i-\theta_i\epsilon_i=p_{i+1}-\theta_i\epsilon_{i+1},\forall i \leq k-1. \label{formula:psysincom1_3}
\end{align}
\end{problem}

By solving Problem \ref{problem:systemgoalincom1}, we can calculate the optimal contract as the following theorem.
\begin{theorem} \label{theorem:incomplete}
In the incomplete information scenario, the optimal contract $\{\epsilon_i^*,p_i^*\}$ is given by
\begin{align*}
\epsilon_i^*&=GH_i^{-\frac{1}{3}}\lambda_i^{\frac{1}{3}}, \\
p_i^*& = \left\{\begin{array}{ll}G(\theta_iH_i^{-\frac{1}{3}}\lambda_i^{\frac{1}{3}}+
\sum_{j=i+1}^k\Delta\theta_jH_j^{-\frac{1}{3}}\lambda_j^{\frac{1}{3}}),
&i\neq k, \\
G\theta_kH_k^{-\frac{1}{3}}\lambda_k^{\frac{1}{3}},
&i=k, \end{array}\right.
\end{align*}
where
\begin{align}
\Delta\theta_i&=\theta_i-\theta_{i-1}, \label{formula:deltatheta}\\
H_i&=\left\{\begin{array}{ll}\lambda_1\theta_1,&i=1, \\
\lambda_i\theta_i+\Delta\theta_i\sum_{j=1}^{i-1}\lambda_j,
&i>1, \end{array}\right. \label{formula:Hi}\\
G&=\frac{B}{\sum_{j=1}^kH_j^{\frac{2}{3}}\lambda_j^{\frac{1}{3}}}.
\end{align}
\end{theorem}

The proof of Theorem \ref{theorem:incomplete} is given in Appendix \ref{appendix:4}. Next, we compare FC's aggregation accuracy under incomplete and complete information scenarios. In Fig. \ref{fig:ratio}, we show the ratio of FC's aggregation accuracy under incomplete information and complete information scenarios when there are three types. $\lambda_1 = 0,50,100,150,200,250$ correspond to the lines from bottom to top, respectively. In this figure, we only show $\lambda_1$ and $\lambda_2$, and $\lambda_3 = N - \lambda_1 - \lambda_2$. Other parameters are $N = 300, B = 1000, \gamma = 10, \delta = 0.9, \theta_1 = 1, \theta_2 = 2, \theta_3 = 3$. The ratio is a function of PUs' realization $\{\lambda_i\}_{i=1}^3$ in three types, which is always larger than or equal to $1$, as FC achieves best aggregation accuracy under complete information scenario. By analyzing Fig. \ref{fig:ratio}, we have the following observation.
\begin{figure}[t]
\begin{center}
\includegraphics[width=0.35\textwidth]{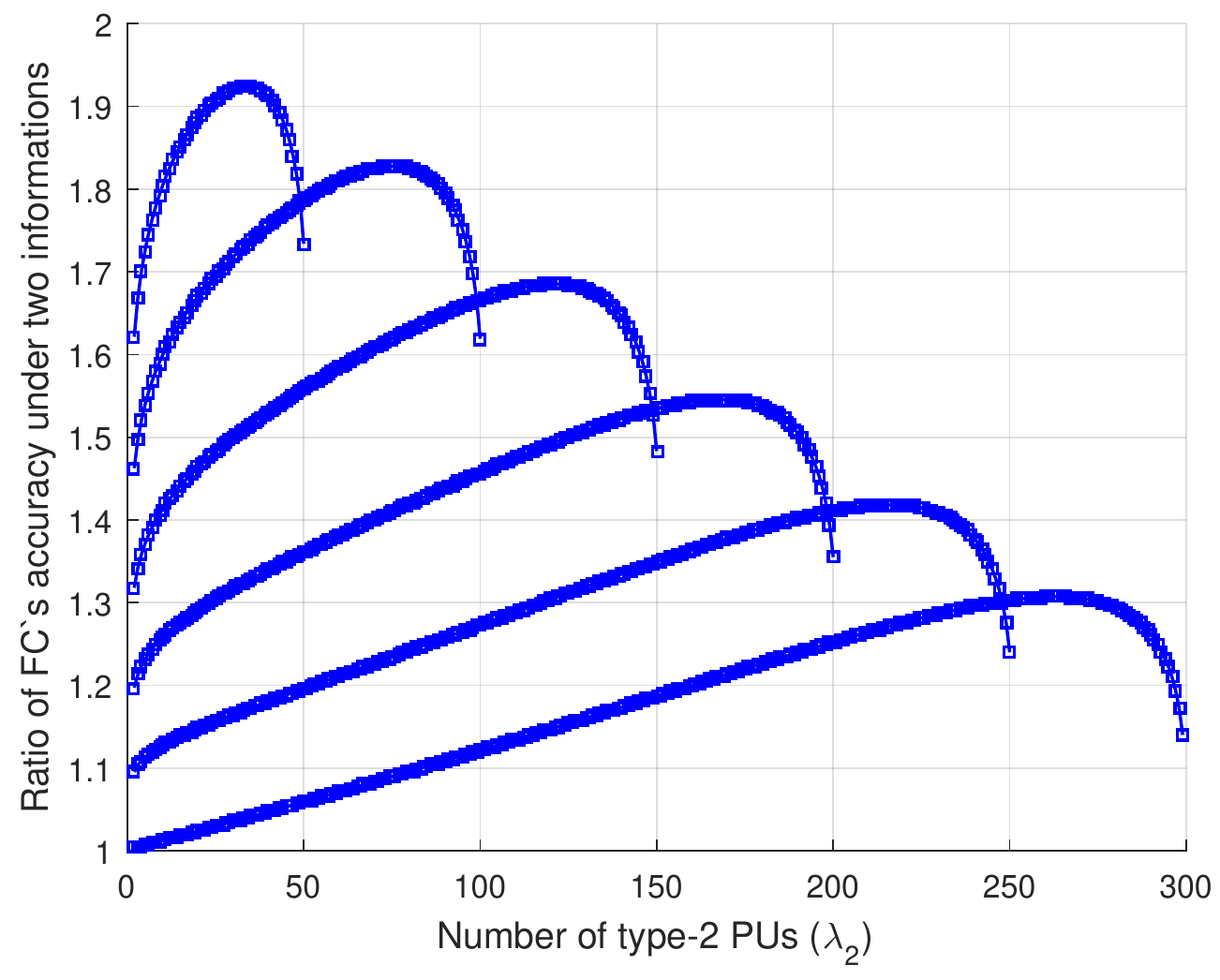}
\end{center}
\caption{The ratio of FC's aggregation accuracy under incomplete information and complete information as a function of PUs' realization in three types, i.e., $\frac{\alpha_I}{\alpha_C}$.}
\label{fig:ratio}
\end{figure}

\begin{observation}
Compared with complete information, FC achieves worse aggregation accuracy, i.e., larger $\alpha$, under incomplete information. The gap between FC's aggregation accuracy under two information scenarios is minimized when all PUs belong to the highest type, i.e., type-$3$. For fixed $\lambda_1$, the gap increases when the number of type-$3$ PUs decrease until they reach a small value.
\end{observation}

The ratio reaches $1$ when all PUs belong to the highest type, since in this situation, all PUs obtains zero utilities as in the complete information scenario. When the number of type-$3$ PUs decrease, the information loss increase, which lead to the increase of the gap. However, when the number of type-$3$ PUs reach a small value, the effect of information loss decreases compared to the complete information, so that the ratio increase.

\subsection{Discussions on Practical Implementation}
By solving the above optimization problem, we could provide a menu of optimal contracts to incentivize all types of PUs' participation in crowdsensing. However, PUs' action, if cannot be monitored by FC, may deviate from the contract in practice, e.g., a selfish PU may add noise with higher PPL than which signed in the contract to achieve higher utility. To ensure that all PUs generate noise strictly with the PPLs signed in the contract, we need a trusted app installed in the mobile device \cite{zhuo2015privacy}. Once the contract is signed, the noise level would be controlled by the trusted app, whose PPL can be monitored by FC.



\section{Generalization to the Continuous Case}\label{section:continuum}
In this section, we will analyze the optimal contract design when PUs' types are continuous.

We assume that PUs' types $\theta$ are in the interval $[\underline{\theta},\overline{\theta}]$, and the probability density function of $\theta$ is $h(\theta)$. Similar to the analysis in the discrete case, FC can design the optimal contracts by solving the following optimization problem:
\begin{problem}\label{problem:continuous}
\begin{align}
&\min\ \int_{\underline{\theta}}^{\overline{\theta}}\frac{h(\theta)}{\epsilon^2(\theta)}d\theta, \nonumber \\
&s.t.\ \ \int_{\underline{\theta}}^{\overline{\theta}}p(\theta)h(\theta)d\theta\leq B, \label{formula:continuous_budget} \\
& \ \ \ \ \ p(\theta)-\theta\epsilon(\theta)\geq0, \label{formula:continuous_IR} \\
& \ \ \ \ \ p(\theta)-\theta\epsilon(\theta)\geq p(\hat{\theta})-\theta\epsilon(\hat{\theta}), \forall \hat{\theta}\neq\theta. \label{formula:continuous_IC}
\end{align}
\end{problem}
where (\ref{formula:continuous_budget}) is the budget constraint, (\ref{formula:continuous_IR}) is the IR constraints and (\ref{formula:continuous_IC}) is the IC constraints.

Notice that the IR and IC constraints in (\ref{formula:continuous_IR}) and (\ref{formula:continuous_IC}) are infinite since $\theta$ is a continuous value. The infinite constraints makes it difficult to solve the optimization problem. Similarly, we first reduce the IR and IC constraints by the following two lemmas.
\begin{lemma}
The infinite IR constraints can be reduced to the following one constraint,
\begin{align}\label{formula:continuous_IR1}
p(\overline{\theta})-\overline{\theta}\epsilon(\overline{\theta})=0.
\end{align}
\end{lemma}

\begin{proof}
We can derive the following inequalities based on the IC constraints,
\begin{align*}
p(\theta)-\theta\epsilon(\theta)&\geq p(\overline{\theta})-\theta\epsilon(\overline{\theta}) \\
&\geq p(\overline{\theta})-\overline{\theta}\epsilon(\overline{\theta}), \forall \theta\neq\overline{\theta}.
\end{align*}

Thus, the IR constraint satisfied for type-$\overline{\theta}$ PUs implies that it satisfied for all $\theta\in[\underline{\theta},\overline{\theta}]$. Then, we can reduce IR constraint to $p(\overline{\theta})-\overline{\theta}\epsilon(\overline{\theta})\geq0$. Moreover, if there exists an optimal contract $(\epsilon(\overline{\theta}),p(\overline{\theta}))$ such that $p(\overline{\theta})-\overline{\theta}\epsilon(\overline{\theta})>0$, we can always find a larger $\epsilon(\overline{\theta})$ to achieve better aggregation until $p(\overline{\theta})-\overline{\theta}\epsilon(\overline{\theta})=0$, which lead to correctness of this lemma.
\end{proof}

\begin{lemma}
The infinite IC constraints can be reduced to the following two constraints,
\begin{align}
\frac{d\epsilon(\theta)}{d\theta}&\leq0, \label{formula:continuous_IC1} \\
\frac{dp(\theta)}{d\theta}&-\theta\frac{d\epsilon(\theta)}{d\theta}=0. \label{formula:continuous_IC2}
\end{align}
\end{lemma}

\begin{proof}
Based on (\ref{formula:continuous_IC}), we can derive the following two local conditions for type-$\theta$ PUs,
\begin{align}
\frac{dp(\hat{\theta})}{d\hat{\theta}}\Big|_{\hat{\theta}=\theta}
-\theta\frac{d\epsilon(\hat{\theta})}{d\hat{\theta}}\Big|_{\hat{\theta}=\theta}=0, \label{formula:firstorder}\\
\frac{d^2p(\hat{\theta})}{d\hat{\theta}^2}\Big|_{\hat{\theta}=\theta}
-\theta\frac{d^2\epsilon(\hat{\theta})}{d\hat{\theta}^2}\Big|_{\hat{\theta}=\theta}\leq0. \label{formula:secondorder}
\end{align}

Since (\ref{formula:firstorder})(\ref{formula:secondorder}) hold for all $\theta\in[\underline{\theta},\overline{\theta}]$, we have
\begin{align}
\frac{dp(\theta)}{d\theta}-\theta\frac{d\epsilon(\theta)}{d\theta}=0, \label{formula:firstorder1}\\
\frac{d^2p(\theta)}{d\theta^2}-\theta\frac{d^2\epsilon(\theta)}{d\theta^2}\leq0. \label{formula:secondorder1}
\end{align}

By differentiating (\ref{formula:firstorder1}), we can simplify (\ref{formula:secondorder1}) as
\begin{align}\label{formula:monotonic}
\frac{d\epsilon(\theta)}{d\theta}\leq0.
\end{align}

Then, we prove that (\ref{formula:firstorder1}) and (\ref{formula:monotonic}) hold globally. By integrating (\ref{formula:firstorder1}) from $\hat{\theta}$ to $\theta$, we have
\begin{align} \label{formula:global}
p(\theta)-p(\hat{\theta})=\theta\epsilon(\theta)-\theta\epsilon(\hat{\theta})
-\int_{\hat{\theta}}^\theta\epsilon(u)du.
\end{align}

Rearrange (\ref{formula:global}), we have
\begin{align*}
p(\theta)-\theta\epsilon(\theta)=p(\hat{\theta})-\hat{\theta}\epsilon(\hat{\theta})+
(\hat{\theta}-\theta)\epsilon(\hat{\theta})-\int_{\hat{\theta}}^\theta\epsilon(u)du.
\end{align*}

Since $\epsilon(\theta)$ is non-increasing, we have $(\hat{\theta}-\theta)\epsilon(\hat{\theta})-\int_{\hat{\theta}}^\theta\epsilon(u)du\geq0$. Thus, we can conclude that $p(\theta)-\theta\epsilon(\theta)\geq p(\hat{\theta})-\hat{\theta}\epsilon(\hat{\theta})$ for all $\hat{\theta}\neq\theta$, which indicate that (\ref{formula:firstorder1}) and (\ref{formula:monotonic}) hold globally.
\end{proof}

Similar to the analysis in the discrete case, the budget constraint (\ref{formula:continuous_budget}) can take the equal sign, i.e.,
\begin{align}\label{formula:continuous_budget1}
\int_{\underline{\theta}}^{\overline{\theta}}p(\theta)h(\theta)d\theta=B.
\end{align}

Then, we can transform Problem \ref{problem:continuous} to the following problem:
\begin{problem}\label{problem:continuous1}
\begin{align*}
&\min\ \int_{\underline{\theta}}^{\overline{\theta}}\frac{h(\theta)}{\epsilon^2(\theta)}d\theta, \nonumber \\
&s.t.\ \ (\ref{formula:continuous_budget1})(\ref{formula:continuous_IR1})(\ref{formula:continuous_IC1})(\ref{formula:continuous_IC2}).
\end{align*}
\end{problem}

Notice that Problem \ref{problem:continuous1} is a functional extreme value problem, we can utilize the optimal control method to solve this problem.

Let $u(\theta)=\epsilon(\theta)$ be the control variable, and let $x_1(\theta)=p(\theta)-\theta\epsilon(\theta)$ be the state variable. Then, we have
\begin{align*}
\dot{x}_1(\theta)&=\dot{p}(\theta)-\epsilon(\theta)-\theta\dot{\epsilon}(\theta) \\
&=-\epsilon(\theta)=-u(\theta),
\end{align*}
where the second equality is due to (\ref{formula:continuous_IC2}).

To deal with the budget constraint (\ref{formula:continuous_budget}), we can define a new state variable 
\begin{align}
\dot{x}_2(\theta)=p(\theta)h(\theta)=[x_1(\theta)+\theta u(\theta)]h(\theta)
\end{align}

Based on (\ref{formula:continuous_budget}), we can derive the following transversality condition,
\begin{align}\label{formula:transversality}
x_2(\overline{\theta})-x_2(\underline{\theta})=B.
\end{align}

Thus, the Hamiltonian of the optimal control problem is
\begin{align*}
&H[x(\theta),u(\theta),\lambda(\theta),\theta] \\
=&\frac{h(\theta)}{u^2(\theta)}-\lambda_1(\theta)u(\theta)+\lambda_2(\theta)[x_1(\theta)+\theta u(\theta)]h(\theta),
\end{align*}
where $\lambda_1(\theta)$ and $\lambda_2(\theta)$ is the co-state variables.

According to the Euler-Lagrange equation for optimal control problem, we have the following conditions,
\begin{align*}
\frac{\partial{H}}{\partial{u}}&=\frac{-2h(\theta)}{u^3(\theta)}-\lambda_1+\lambda_2\theta h(\theta)=0, \\
\dot{\lambda}_1(\theta)&=-\frac{\partial{H}}{\partial{x_1}}=-\lambda_2h(\theta), \\
\dot{\lambda}_2(\theta)&=-\frac{\partial{H}}{\partial{x_2}}=0.
\end{align*}

Thus, we can calculate the co-state variables as,
\begin{align*}
\lambda_2(\theta)&=c_1, \\
\lambda_1(\theta)&=-c_1H(\theta)+c_2,
\end{align*}
where $c_1$ and $c_2$ are constants which can be calculated by the transversality conditions (\ref{formula:transversality})(\ref{formula:continuous_IR1}).

Then, we the optimal contract $[\epsilon^*(\theta),p^*(\theta)]$ is given by,
\begin{align*}
\epsilon^*(\theta)
&=u^*(\theta) \\
&=\sqrt[3]{\frac{2h(\theta)}{c_1\theta h(\theta)-c_1H(\theta)-c_2}}, \\
p^*(\theta)
&=x_1(\theta)+\theta\epsilon(\theta) \\
&=\theta\epsilon^*(\theta)-\int_{\underline{\theta}}^{\theta}\epsilon^*(\tau)d\tau.
\end{align*}


\section{Simulation Studies}\label{section:simulation}
In this section, we first validate the feasibility of the proposed contracts, and then analyze the impact of different system parameters on the aggregation accuracy.
\begin{table}[!htp]
	\caption{Simulation settings}
	\centering
	\begin{tabular}[c]{|c|c|c|}
		\hline \label{table:simulationsettings}
		Parameter            			& \multicolumn{2}{c|}{Value}   \\ \hline\hline
		Number of PUs ($N$)				& \multicolumn{2}{c|}{$200$}     \\ \hline
		Privacy preference ($\theta$) 	& \multicolumn{2}{c|}{$[5,15]$}  \\ \hline
		Number of PUs' types 			& Feasibility & 20             \\ \cline{2-3} 
		($k$)                			& Performance & $[5,20]$         \\ \hline
		Budget constraint    			& Feasibility & $1000$           \\ \cline{2-3} 
		($B$)                			& Performance & $[500,1000]$     \\ \hline
	\end{tabular}
\end{table}


The simulation settings are shown in Table \ref{table:simulationsettings}. We assume there are $200$ PUs and their privacy preferences are from $5$ to $10$ in both simulations. For simplicity, we consider a uniform distribution of PUs' privacy preference. To illustrate the feasibility of the proposed contract, we set the number of PUs' types $k$ and the budget constraint $B$ to $20$ and $1000$ respectively. To evaluate the impact of parameter $k$ and $B$ to the aggregation accuracy, we set the value ranges to $[5,20]$ and $[500,1000]$ respectively. 

\subsection{Contract Feasibility}
In this subsection, we illustrate that the proposed optimal contracts satisfy both the \emph{monotonic} property and \emph{incentive compatibility} property.

\begin{figure}[ht]
\begin{center}
\includegraphics[width=0.35\textwidth]{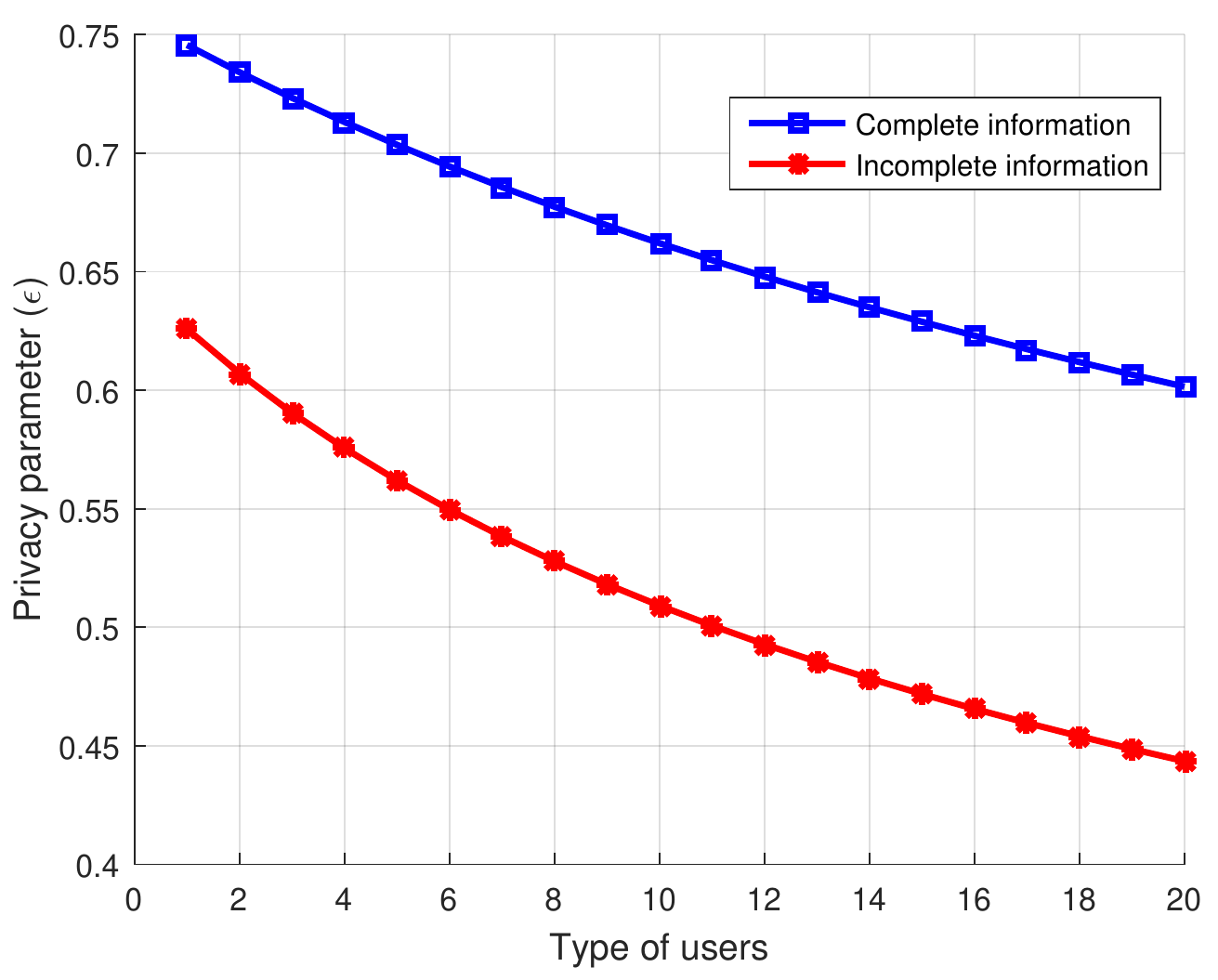}
\end{center}
\caption{Contract monotonicity.}\label{fig:monotonic}
\end{figure}

\begin{figure}[ht]
\begin{center}
\includegraphics[width=0.35\textwidth]{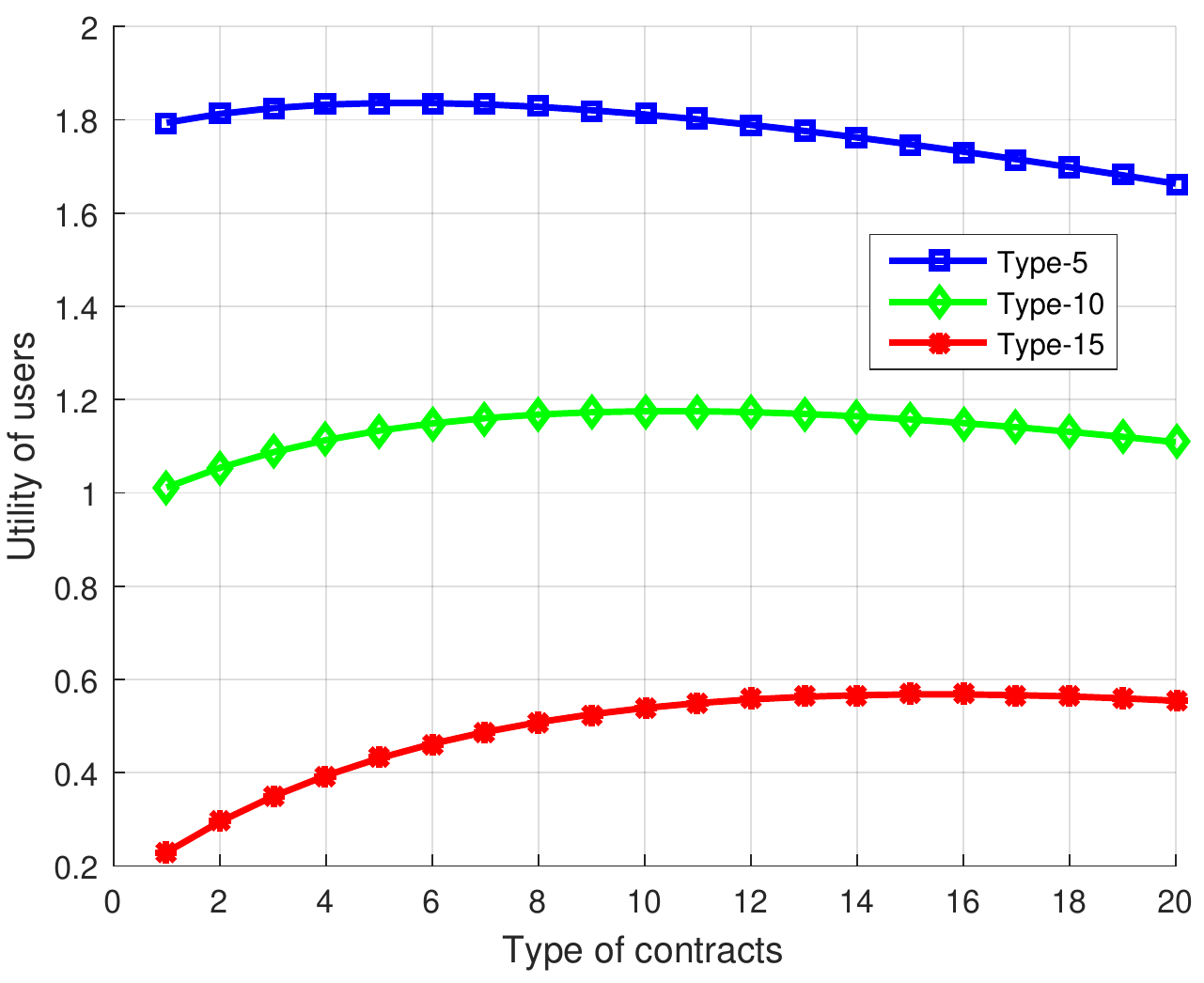}
\end{center}
\caption{Contract incentive compatibility.}\label{fig:IC}
\end{figure}

\begin{figure}[ht]
\begin{center}
\includegraphics[width=0.35\textwidth]{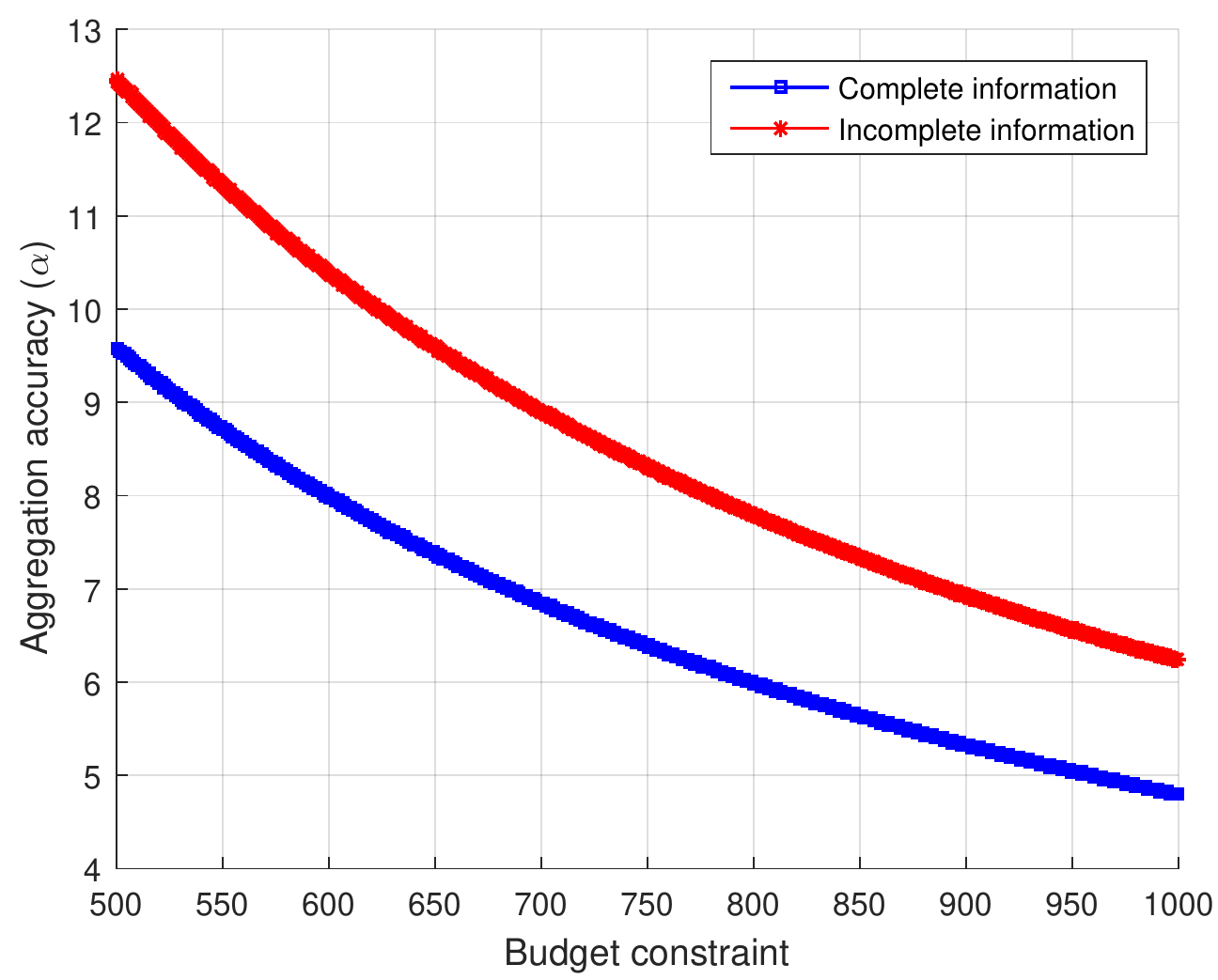}
\end{center}
\caption{Aggregation accuracy Vs. budget constraint.} \label{fig:discrete_performance_B}
\end{figure}

\begin{figure}[ht]
\begin{center}
\includegraphics[width=0.35\textwidth]{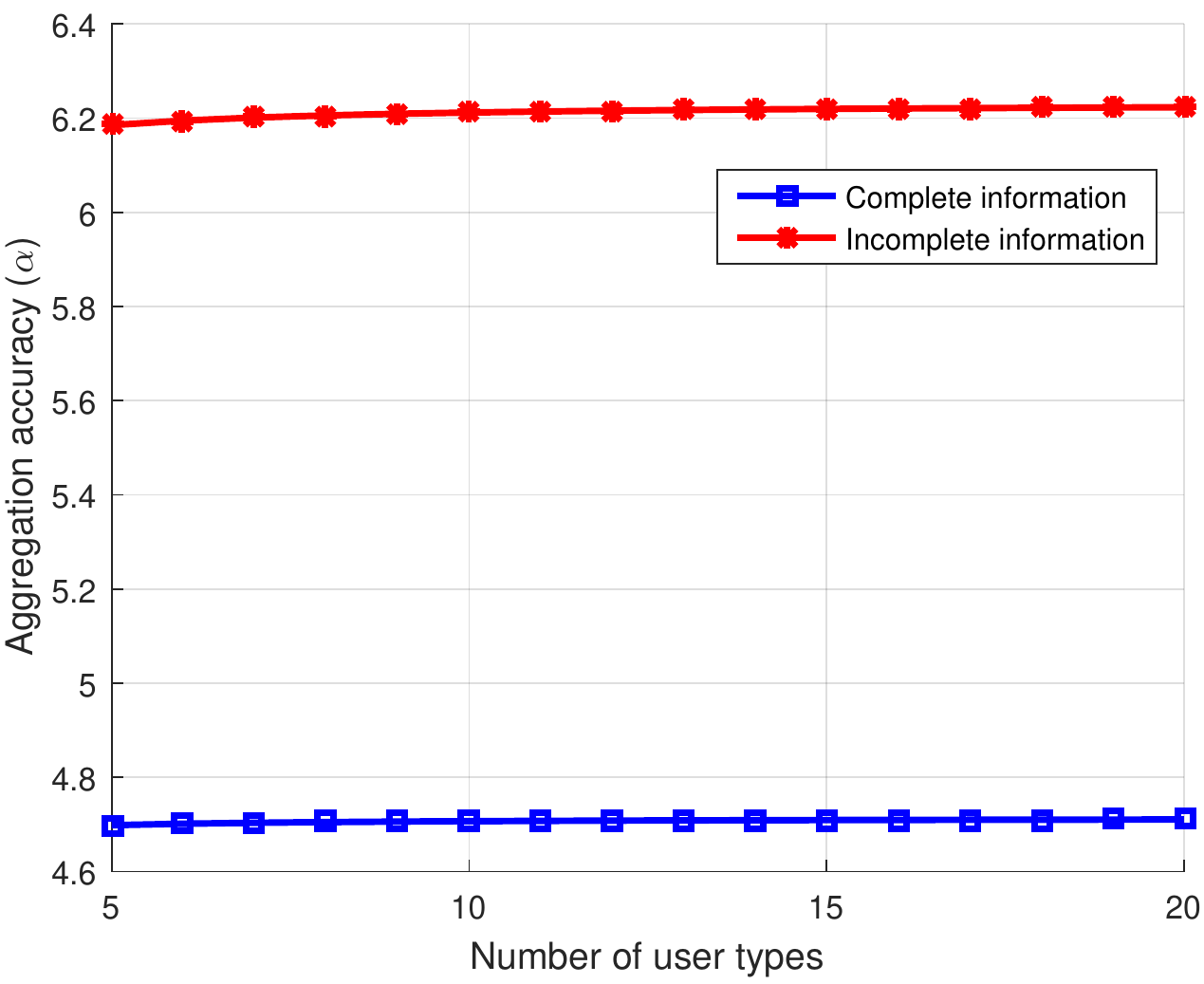}
\end{center}
\caption{Aggregation accuracy Vs. number of user types.} \label{fig:discrete_performance_k}
\end{figure}

Fig. \ref{fig:monotonic} shows that $\epsilon$ decreases when PUs' types increase. Since a smaller $\epsilon$ means higher PPL, Fig. \ref{fig:monotonic} indicates that PUs with higher type tend to choose higher PPL, which validate the \emph{monotonic} property discussed in Lemma \ref{lemma:monotonic}. Besides, the result is accord with our intuition that the FC choose to buy less privacy from the PUs with higher privacy preference to reduce the payment. In another hand, we find that under the same budget constraint, PUs' PPLs under complete information scenario are lower than which in incomplete information scenario. The reason is that in complete information scenario, the FC knows each PU's precise type, so that the contract designed to all types of PUs can take zero utilities, as Lemma \ref{lemma:equality} shows. However, in the incomplete information scenario, PUs' precise types is unknown to the system. Thus, only the highest type contract can take zero utility, whereas other types of PUs' utilities should remain strictly positive, since otherwise, the PUs with lower type will pretend to be higher type to achieve higher utility.

In Fig. \ref{fig:IC}, we show the utility function of type-$5$, type-$10$ and type-$15$ PUs when they choose all types of contracts. Notice that the utility function is concave for all types of PUs, and each type of PUs achieve their optimal utilities when they choose the corresponding contract, e.g., type-$5$ PUs achieve their optimal utilities when they choose type-$5$ contract, which validate the \emph{incentive compatibility} property. Additionally, we observe that the PUs with lower type can achieve higher utility when they choose the same contract. The reason is that the lower type PUs have lower privacy preference $\theta_i$, according to PUs' utility definition $u_j=p_j-\theta_i\epsilon_j,\forall j$, a smaller $\theta_i$ result in higher utility.

\subsection{System Performance}
In this subsection, we show the impact of different system parameters on the aggregation accuracy.

Fig. \ref{fig:discrete_performance_B} shows the impact of the amount of budget on the aggregation accuracy when other parameters are fixed. We observe that $\alpha$ decreases when the amount of budget increases. Since a smaller $\alpha$ means a better aggregation accuracy, \ref{fig:discrete_performance_B} indicates that larger amount of budget lead to better aggregation accuracy. The reason is obvious, when the FC possesses more budget, it can provide more incentive to drive PUs to choose lower PPL, which lead to better aggregation accuracy.

In Fig. \ref{fig:discrete_performance_k}, we evaluate the impact of number of PUs' types on the aggregation accuracy when other parameters are fixed. Fig. \ref{fig:discrete_performance_k} shows that, the aggregation accuracy decreases with the number of PUs' types. Recall the reduced IR constraint $p_k - \theta_k\epsilon_k = 0$ and IC constraints $p_i - \theta_i\epsilon_i = p_{i+1} - \theta_i\epsilon_{i+1}$, we can set the utility of higher type PUs more close to $0$, which means less additional payments. Thus, the increase of PUs' types lead to more additional payment which decrease the aggregation accuracy under budget constraint.

\section{Conclusion}\label{section:conclusion}
In this paper, we designed an incentive mechanism REAP to compensate PUs' privacy loss. Unlike previous mechanisms, we did not require FC to be trustworthy and allow PUs to add well calibrated noise to their sensing data before reporting them to FC. Then, in order to achieve better aggregation accuracy under a budget constraint, we devised a contract-based incentive mechanism to induce PUs to lower down their PPL. Optimal contracts with closed form were derived in both complete and incomplete information scenarios. Our results were generalized to the continuous case. Extensive simulations were conducted to validate the feasibility of our proposed incentive mechanism.


\bibliographystyle{IEEEtran}
\bibliography{cited}

\appendix
\section{Proof of Lemma 3.2}\label{appendix:1} 
\begin{proof}
The aggregation error of the randomized sensing data can be written as
\begin{align*}
\hat{s}-s=\frac{1}{n}\sum_{i=1}^N(d_i+\eta_i) - \frac{1}{n}\sum_{i=1}^Nd_i
=\frac{1}{n}\sum_{i=1}^N\eta_i.
\end{align*}

Recall that the variance of Laplacian random variable $\eta_i\sim Lap(0,b_i)$ is $2b_i^2$, i.e., $D(\eta_i)=2b_i^2$, we can derive that
\begin{align*}
D(\frac{1}{n}\sum_{i=1}^N\eta_i)=\frac{2}{n^2}\sum_{i=1}^nb_i^2.
\end{align*}

Therefore, from the \emph{Chebyshev's inequality}, we have
\begin{align*}
P[|s-\hat{s}| \geq \alpha] \leq \frac{2}{\alpha^2n^2}\sum_{i=1}^{n}b_i^2,
\end{align*}
which indicates that the aggregated randomized sensing data satisfies $(\alpha, \frac{2}{\alpha^2n^2}\sum_{i=1}^{n}b_i^2)$-accuracy.

Thus, for a given confidence level $\delta \leq 1$, we have
\begin{align*}
\alpha = \frac{\sqrt{2}\gamma}{n\sqrt{1-\delta}}\sqrt{\sum_{i=1}^nb_i^2}
\end{align*}

Substituting $b_i=\frac{\gamma_i}{\epsilon_i}$ into the above formula, and set $\gamma_i=\gamma$ for all $i$, we derive
\begin{align*}
\alpha = \frac{\sqrt{2}\gamma}{n\sqrt{1-\delta}}\sqrt{\sum_{i=1}^n\frac{1}{\epsilon_i^2}}.
\end{align*}
\end{proof}

\section{Proof of Theorem 4.2}\label{appendix:2} 
\begin{proof}
Substituting (\ref{formula:IR1}) to (\ref{formula:budget1}), we have
\begin{align}\label{formula:budget_epsilon}
\sum_{i=1}^k\lambda_i\theta_i\epsilon_i=B
\end{align}

The Lagrangian of Problem \ref{problem:systemgoalcom1} can be written as
\begin{align*}
L(\epsilon_i,\alpha)=\sum_{i=1}^k[\frac{\lambda_i}{\epsilon_i^2}+\alpha\lambda_i\theta_i\epsilon_i]-\alpha B,
\end{align*}
where $\alpha$ is the Lagrangian multiplier.

Based on the KKT condition, we have
\begin{align*}
\frac{\partial L}{\partial\epsilon_i}=\frac{-2\lambda_i}{\epsilon_i^3}+\alpha\lambda_i\theta_i=0,\ \ \ \forall i.
\end{align*}

Solving the above equation obtain $\epsilon_i=\sqrt[3]{\frac{2}{\alpha}}\theta_i^{-\frac{1}{3}}$. Substituting this formula to (\ref{formula:budget_epsilon}), we have
\begin{align*}
\sqrt[3]{\frac{2}{\alpha}}=\frac{B}{\sum_{i=1}^k\lambda_i\theta_i^{\frac{2}{3}}}.
\end{align*}

Therefore, $\epsilon_i^*$ is given by 
\begin{align}
\epsilon_i^*&=\frac{B}{\sum_{j=1}^k\lambda_j\theta_j^{\frac{2}{3}}}\theta_i^{-\frac{1}{3}}, \label{formula:optcontractcomepsilon}
\end{align} 

Substituting (\ref{formula:optcontractcomepsilon}) to $p_i^*-\theta_i\epsilon_i^*=0$, $p_i^*$ can be calculated as
\begin{align}
p_i^*&=\frac{B}{\sum_{j=1}^k\lambda_j\theta_j^{\frac{2}{3}}}\theta_i^{\frac{2}{3}}. \label{formula:optcontractcomp}
\end{align}
\end{proof}

\section{Proof of Lemma 4.5}\label{appendix:3}
\begin{proof}
We will conduct the proof of this lemma by three steps.

Firstly, we prove that if $p_i-\theta_i\epsilon_i\geq p_{i-1}-\theta_i\epsilon_{i-1}$ satisfies, then $p_i-\theta_i\epsilon_i\geq p_j-\theta_i\epsilon_j$ hold for all $j\in\{i-1,i-2,\cdots,1\}$.

Based on the IC constraint, we have
\begin{align}
p_i-\theta_i\epsilon_i&\geq p_{i-1}-\theta_i\epsilon_{i-1}, \\
p_{i-1}-\theta_{i-1}\epsilon_{i-1}&\geq p_{i-2}-\theta_{i-1}\epsilon_{i-2}. \label{formula:down}
\end{align}

Formula (\ref{formula:down}) can be transformed to the following form
\begin{align*}
\theta_{i-1}(\epsilon_{i-2}-\epsilon_{i-1})\geq p_{i-2}-p_{i-1}.
\end{align*}

Recall the monotonic property in Lemma \ref{lemma:monotonic}, we know that $\theta_{i-1}\leq\theta_i$ and $\epsilon_{i-2}\geq\epsilon_{i-1}$. Thus, we have $\theta_i(\epsilon_{i-2}-\epsilon_{i-1})\geq p_{i-2}-p_{i-1}$ or $p_{i-1}-\theta_i\epsilon_{i-1}\geq p_{i-2}-\theta_i\epsilon_{i-2}$. Following the same step, we have
\begin{align*}
p_i-\theta_i\epsilon_i&\geq p_{i-1}-\theta_i\epsilon_{i-1}
\geq\cdots
\geq p_1-\theta_i\epsilon_1.
\end{align*}

These inequalities lead to the correctness of this step.

Secondly, we prove that if $p_i-\theta_i\epsilon_i\geq p_{i+1}-\theta_i\epsilon_{i+1}$ satisfies, then $p_i-\theta_i\epsilon_i\geq p_j-\theta_i\epsilon_j$ hold for all $j\in\{i+1,i+2,\cdots,k\}$.

Similar to the proof of the first step, we have
\begin{align*}
p_i-\theta_i\epsilon_i&\geq p_{i+1}-\theta_i\epsilon_{i+1}
\geq\cdots
\geq p_1-\theta_i\epsilon_1,
\end{align*}
which lead to the correctness of this step. Notice that for an optimal contract, we have $p_i-\theta_i\epsilon_i=p_{i+1}-\theta_i\epsilon_{i+1}$ holds, since otherwise, we can always find a larger $\epsilon_i$ to achieve a better aggregation accuracy until the equal signs hold.

Thirdly, we prove that $p_i-\theta_i\epsilon_i= p_{i+1}-\theta_i\epsilon_{i+1}$ implies $p_i-\theta_i\epsilon_i\geq p_{i-1}-\theta_i\epsilon_{i-1}$.

It is obvious that $\theta_i(\epsilon_{i-1}-\epsilon_i)\geq\theta_{i-1}(\epsilon_{i-1}-\epsilon_i)$, rearrange this inequality, we have
\begin{align*}
p_i-\theta_i\epsilon_i\geq p_i+\theta_{i-1}\epsilon_{i-1}-\theta_{i-1}\epsilon_i-\theta_i\epsilon_{i-1}.
\end{align*}

Since $p_i-\theta_i\epsilon_i= p_{i+1}-\theta_i\epsilon_{i+1}$, then $p_{i-1}-\theta_{i-1}\epsilon_{i-1}= p_i-\theta_{i-1}\epsilon_i$ hold, i.e., $p_i+\theta_{i-1}\epsilon_{i-1}-\theta_{i-1}\epsilon_i=p_{i-1}$. Thus, we have $p_i-\theta_i\epsilon_i\geq p_{i-1}-\theta_i\epsilon_{i-1}$.

In summary, $p_i-\theta_i\epsilon_i= p_{i+1}-\theta_i\epsilon_{i+1}$ implies $p_i-\theta_i\epsilon_i\geq p_j-\theta_i\epsilon_j,\forall j\neq i$, which end the proof of this lemma.
\end{proof}

\section{Proof of Theorem 4.6}\label{appendix:4}
\begin{proof}
Based on (\ref{formula:psysincom1_2}) and (\ref{formula:psysincom1_3}), we have
\begin{align}
p_{k-1}-\theta_{k-1}\epsilon_{k-1}&=p_k-\theta_{k-1}\epsilon_k \nonumber \\
&=\theta_k\epsilon_k-\theta_{k-1}\epsilon_k \nonumber \\
&=(\theta_k-\theta_{k-1})\epsilon_k \label{formula:recursive}
\end{align}

Let $\Delta\theta_k=\theta_k-\theta_{k-1}$, we can rewrite (\ref{formula:recursive}) as $p_{k-1}=\theta_{k-1}\epsilon_{k-1}+\Delta\theta_k\epsilon_k$.

Following the same procedure, we can conclude that
\begin{align}\label{formula:pi}
p_i=\left\{\begin{array}{ll}\theta_i\epsilon_i+\sum_{j=i+1}^k\Delta\theta_j\epsilon_j,&i\neq k, \\
\theta_k\epsilon_k,&i=k,\end{array}\right.
\end{align}
where $\Delta\theta_i$ is defined by (\ref{formula:deltatheta}).

Then, we have
\begin{align*}
\sum_{i=1}^k\lambda_ip_i
=&\sum_{i=1}^{k-1}[\lambda_i\theta_i\epsilon_i+\lambda_i\sum_{j=i+1}^k\Delta\theta_j\epsilon_j]
+\lambda_k\theta_k\epsilon_k \\
=&\lambda_k\theta_k\epsilon_k +\lambda_{k-1}\theta_{k-1}\epsilon_{k-1}+\lambda_{k-1}\Delta\theta_k\epsilon_k \\
&+\lambda_{k-2}\theta_{k-2}\epsilon_{k-2}+\lambda_{k-2}[\Delta\theta_{k-1}\epsilon_{k-1}+\Delta\theta_k\epsilon_k] \\
&\vdots \\
&+\lambda_1\theta_1\epsilon_1+\lambda_1[\Delta\theta_2\epsilon_2+\cdots+\Delta\theta_k\epsilon_k] \\
=&\epsilon_k[\lambda_k\theta_k+\Delta\theta_k(\lambda_{k-1}+\cdots+\lambda_1)] \\
&+\epsilon_{k-1}[\lambda_{k-1}\theta_{k-1}+\Delta\theta_{k-1}(\lambda_{k-2}+\cdots+\lambda_1)] \\
&\vdots \\
&+\epsilon_1\lambda_1\theta_1.
\end{align*}

Rearrange the above equation by $\epsilon_i$, we can get
\begin{align}\label{formula:epsilon}
\sum_{i=1}^k\lambda_ip_i=\sum_{i=1}^kH_i\epsilon_i=B,
\end{align}
where $H_i$ is defined by (\ref{formula:Hi}).

Thus, the Lagrangian of Problem \ref{problem:systemgoalincom1} is
\begin{align*}
L(\epsilon,\alpha)=\sum_{i=1}^k[\frac{\lambda_i}{\epsilon_i^2}+\alpha H_i\epsilon_i]-\alpha B,
\end{align*}
where $\alpha$ is the Lagrangian multiplier.

Based on the KKT condition, we have
\begin{align*}
\frac{\partial{L}}{\partial{\epsilon_i}}=\frac{-2\lambda_i}{\epsilon_i^3}+\alpha H_i=0.
\end{align*}

Then, we can calculate $\epsilon_i$ as 
\begin{align}\label{formula:epsilon_i}
\epsilon_i=\sqrt[3]{\frac{2}{\alpha}}(\frac{\lambda_i}{H_i})^{\frac{1}{3}}
\end{align}
Substituting (\ref{formula:epsilon_i}) to (\ref{formula:epsilon}), we obtain 
\begin{align*}
\sqrt[3]{\frac{2}{\alpha}}=\frac{B}{\sum_{j=1}^kH_j^{\frac{2}{3}}\lambda_j^{\frac{1}{3}}}
\end{align*}

Thus, the optimal contract $\epsilon_i^*$ is given by
\begin{align}\label{formula:epsilonistar}
\epsilon_i^*=\frac{B}{\sum_{i=1}^kH_i^{\frac{2}{3}}\lambda_i^{\frac{1}{3}}}
H_i^{-\frac{1}{3}}\lambda_i^{\frac{1}{3}}
\end{align}

Then, we can calculate the $k$-th contract as,
\begin{align*}
p_k^*=\theta_k\epsilon_k^*=\frac{B}{\sum_{j=1}^kH_j^{\frac{2}{3}}\lambda_j^{\frac{1}{3}}}
\theta_kH_k^{-\frac{1}{3}}\lambda_k^{\frac{1}{3}}.
\end{align*}

Substitute (\ref{formula:epsilonistar}) to (\ref{formula:pi}) and rearrange, we can achieve other contracts when $i\neq k$.
\end{proof}


\end{document}